\documentclass[
notitlepage
,floatfix,
aps,
prl,
reprint,
superscriptaddress,
]{revtex4-1}
\usepackage[utf8]{inputenc}

\usepackage[caption=false]{subfig}
\usepackage{graphicx} 
\usepackage{epstopdf}
\usepackage{array}
\usepackage{verbatim}
\usepackage{amsmath,amsfonts,amssymb,amscd,mathtools}
\usepackage{amsthm}
\usepackage{tabularx}
\usepackage{stmaryrd}
\usepackage{enumerate}
\usepackage[ruled]{algorithm2e}
\usepackage{wasysym}
\usepackage{braket}
\usepackage{mathrsfs}
\usepackage{microtype}
\usepackage{hyperref}
\usepackage[dvipsnames]{xcolor}
\hypersetup{
    bookmarksnumbered=true, 
    unicode=false, 
    pdfstartview={FitH}, 
    pdftitle={}, 
    pdfauthor={}, 
    pdfsubject={}, 
    pdfcreator={}, 
    pdfproducer={}, 
    pdfkeywords={}, 
    pdfnewwindow=true, 
    colorlinks=true, 
    linkcolor=NavyBlue, 
    citecolor=NavyBlue, 
    filecolor=NavyBlue, 
    urlcolor=NavyBlue 
}

\usepackage{newtxtext,newtxmath}

\theoremstyle{plain}
\newtheorem{thm}{Theorem}
\newtheorem{cor}[thm]{Corollary}
\newtheorem{lem}[thm]{Lemma}
\newtheorem{pro}[thm]{Proposition}

\theoremstyle{definition}

\newcommand{\eq}[1]{(\hyperref[eq:#1]{\ref*{eq:#1}})}

\renewcommand{\sec}[1]{\hyperref[sec:#1]{Section~\ref*{sec:#1}}}
\newcommand{\thrm}[1]{\hyperref[thrm:#1]{Theorem~\ref*{thrm:#1}}}
\newcommand{\lemm}[1]{\hyperref[lemm:#1]{Lemma~\ref*{lemm:#1}}}
\newcommand{\prop}[1]{\hyperref[prop:#1]{Proposition~\ref*{prop:#1}}}
\newcommand{\corr}[1]{\hyperref[corr:#1]{Corollary~\ref*{corr:#1}}}
\newcommand{\fig}[1]{\hyperref[fig:#1]{~\ref*{fig:#1}}}
\newcommand{\deff}[1]{\hyperref[deff:#1]{~\ref*{deff:#1}}}

\newcommand{\mA}{\mathcal{A}}

\newcommand{\mE}{\mathcal{E}}
\newcommand{\mN}{\mathcal{N}}

\newcommand{\mT}{\mathcal{T}}

\newcommand{\mF}{\mathcal{F}}
\newcommand{\mL}{\mathcal{L}}

\newcommand{\mM}{\mathcal{M}}
\newcommand{\mO}{\mathcal{O}}
\newcommand{\mP}{\mathcal{P}}

\newcommand{\mS}{\mathcal{S}}

\newcommand{\mbI}{\mathbb{I}}

\newcommand{\mfD}{\mathfrak{D}}
\newcommand{\mfF}{\mathfrak{F}}

\DeclareMathAlphabet{\matheu}{U}{eus}{m}{n}

\DeclareMathOperator{\Tr}{Tr}
\DeclareMathOperator{\id}{id}

\newcommand{\ketbra}[2]{|{#1}\rangle\!\langle{#2}|}

\newcommand{\ba}{\begin{eqnarray}}
\newcommand{\ea}{\end{eqnarray}}
\newcommand{\bann}{\begin{eqnarray*}}
\newcommand{\eann}{\end{eqnarray*}}
\newcommand{\bal}{\begin{equation}\begin{aligned}}
\newcommand{\eal}{\end{aligned}\end{equation}}
\newcommand{\dm}[1]{\ketbra{#1}{#1}}

\newcolumntype{L}[1]{>{\raggedright}p{#1}}
\newcolumntype{C}[1]{>{\centering}p{#1}}
\newcolumntype{R}[1]{>{\raggedleft}p{#1}}
\newcolumntype{D}{>{\centering\arraybackslash}X}

\newcommand{\sbar}{\;\rule{0pt}{9.5pt}\right|\;}

\newcommand{\lset}{\left\{\left.}
\newcommand{\rset}{\right\}}

\begin{document}

\title{Application of the Resource Theory of Channels to Communication Scenarios}

\author{Ryuji Takagi}
\email{rtakagi@mit.edu}
\affiliation{Center for Theoretical Physics and Department of Physics, Massachusetts Institute of Technology, Cambridge, MA 02139, USA}

\author{Kun Wang}
\email{wangk36@sustech.edu.cn}
\affiliation{Shenzhen Institute for Quantum Science and Engineering, Southern University of Science and Technology, Shenzhen 518055, China}
\affiliation{Center for Quantum Computing, Peng Cheng Laboratory, Shenzhen 518000, China}

\author{Masahito Hayashi}
\email{masahito@math.nagoya-u.ac.jp}
\affiliation{Graduate School of Mathematics, Nagoya University, Nagoya, 464-8602, Japan}
\affiliation{Shenzhen Institute for Quantum Science and Engineering, Southern University of Science and Technology, Shenzhen 518055, China}
\affiliation{Centre for Quantum Technologies, National University of Singapore, 3 Science Drive 2, 117542, Singapore}
\affiliation{Center for Quantum Computing, Peng Cheng Laboratory, Shenzhen 518000, China}

\begin{abstract}
 We introduce a resource theory of channels relevant to communication via quantum channels, in which the set of constant channels --- useless channels for communication tasks --- is considered as the free resource. We find that our theory with such a simple structure is useful to address central problems in quantum Shannon theory --- in particular, we provide a converse bound for the one-shot non-signalling assisted classical capacity that naturally leads to its strong converse property, as well as obtain the one-shot channel simulation cost with non-signalling assistance. We clarify an intimate connection between the non-signalling assistance and our formalism by identifying the non-signalling assisted channel coding with the channel transformation under the maximal set of resource non-generating superchannels, providing a physical characterization of the latter. Our results provide new perspectives and concise arguments to those problems, connecting the recently developed fields of resource theories to `classic' settings in quantum information theory and shedding light on the validity of resource theories of channels as effective tools to address practical problems.  
\end{abstract}

\maketitle

\textit{\textbf{Introduction.}}
---
A central problem addressed in quantum Shannon theory is to understand how much of the resources are required to accomplish the desired communication tasks and how one can efficiently use them. 
In general, the idea of distinguishing costly resources and free resources is helpful for articulating the problem to address, and it has been employed in a series of works in quantum information theory.

Once the precious resources are identified for the given setting, one is naturally motivated to consider quantifying and manipulating them in an appropriate manner, which leads to a general framework called \textit{resource theories}~\cite{Chitambar2019resource}. 
The resource theoretic framework has been applied to various kinds of quantities~\cite{Horodecki2009entanglement,Baumgratz2014,Gour2008reference,Brandao2013thermo,gallego_2015,Veitch_2014stab,Takagi2018convex} and has been employed to extract common features shared by a wide class of resources~\cite{brandao_2015,Liu2017,gour_2017,anshu_2017,regula_2018,Takagi2019operational,li_2018,Takagi2019general,Uola2019conic,Liu2019oneshot,Vijayan2019oneshot,Takagi2019universal,Regula2019characterizing}.
Recently, the framework has been extended beyond the consideration of static resources contained in quantum states to dynamic resources attributed to quantum measurements and channels, and it is currently under active investigation~\cite{guerini_2017,skrzypczyk2018robustness,Oszmaniec2019operational,Uola2019conic,Guff2019measurement,Pirandola2017fundamental,bendana_2017,gour_2018-1,wilde_2018,rosset_2018,Diaz2018usingreusing,theurer_2018,li_2018,Zhuang2018non-Gaussian,seddon2019quantifying,Wang2019quantifying,Takagi2019general,Liu2019channel,Liu2019yuan,Bauml2019channel,Gour2019bipartite}.

Although the idea of resource theory has succeeded to provide much insight into the properties of the interested quantities, a common criticism is that the discussion often ends up with a formalistic level not solving other existing problems, apart from a few attempts along this line for several resource theories of states~\cite{howard_2017,Yunger2018isomer,Takagi2018skew,Fang2019nogo}.
In particular, it has been elusive whether the resource theory of channels would be helpful for answering concrete problems at all.

In this work, we take the first step in this direction. We introduce the \textit{resource theory of communication}, a resource theory of channels relevant to communication via quantum channels.  
Unlike many resource theories of channels with underlying state theories~\cite{bendana_2017,wilde_2018,Zhuang2018non-Gaussian,theurer_2018,Diaz2018usingreusing,li_2018,seddon2019quantifying,Wang2019quantifying,Pirandola2017fundamental}, our setting is not equipped with any theory of states, making the consideration of the resource theories of channels crucial.  
We consider the \textit{generalized robustness of communication} and another related quantity, \textit{max-relative entropy of communication}, as resource quantifiers and provide them with an operational meaning in terms of state discrimination task. 
We set the maximal set of superchannels that do not create resourceful channels as free operations and show that they coincide with the non-signalling assisted channel transformations, which suggests an intimate connection between our framework and the non-signalling assisted communication. 
With this formalism, we address two important problems in quantum Shannon theory: strong converse property and channel simulation cost under resource assistance.

One of the most fundamental properties of quantum channels is the communication capacity, and it is said to have the strong converse property if the error rate necessarily approaches unit whenever the transmission rate exceeds its capacity~\cite{Arimoto1973converse,Verdu1994general,Ogawa1999converse,Winter1999converse,Hayashi2002general,Datta2013smooth,Konig2009strong,Bardhan2015optical,Wilde2014strong,Bennett2014reverse,Berta2011reverse,Gupta2015strong,Fang2018maxinfov2,Morgan2014pretty,Wang2019converse}.
The strong converse property sets a fundamental limitation on information transmission --- one can never hope to succeed in sending information with the rates exceeding the capacity even allowing for a finite error --- and thus provides a sharp notion to the capacity as a phase transition point in the information transmission.
While entanglement assisted channels admit wider choices of encoding than the conventional classical-quantum channels, making the strong converse property appear to be even more subtle and surprising, the strong converse property for the entanglement assisted capacity has been shown by an operational argument by flipping the quantum reverse Shannon theorem~\cite{Bennett2014reverse,Berta2011reverse}. 
Later, a more direct proof of the strong converse property for the entanglement assisted capacity without using the above operational argument but by directly evaluating the error at the rates above the capacity has been shown in Ref.~\cite{Gupta2015strong}, for which several involved techniques were employed.
Here, combining our framework with recent progress in operational characterization of resource theories in terms of discrimination tasks~\cite{Takagi2019operational,Takagi2019general,Uola2019conic,Oszmaniec2019operational}, we provide a concise proof by showing an even stronger claim: the strong converse property for non-signalling assisted capacity, which includes entanglement assisted communication as a special case.
We note that the quantum reverse Shannon theorem under non-signalling assistance~\cite{Fang2018maxinfov2} allows one to prove the strong converse property by an operational argument that is analogous to the one used for the entanglement assisted case~\cite{Bennett2014reverse}.
Thus, our proof can be seen as an alternative proof in a more direct approach in a similar sense of Ref.~\cite{Gupta2015strong}.
Our main result is a converse bound for the one-shot non-signalling assisted classical capacity that immediately leads to the strong converse property together with the recently shown asymptotic equipartition property of max-information of channels~\cite{Fang2018maxinfov2}. (See also Refs.~\cite{Leung2015nonsignal,Duan2016nosignalling,Wang2018converse,Wang2019converse} for related works on the non-signalling assisted capacity.) 

Channel simulation is a reverse task of noiseless communication via noisy channels, in which one is to implement the given noisy channel from the noiseless channel using accessible free resources~\cite{Bennett2014reverse,Berta2011reverse,Fang2018maxinfov2,Berta2013cost,wilde_2018,Pirandola2017fundamental}.
Recently, the one-shot channel simulation cost with non-signalling assistance has been obtained by techniques based on the semidefinite programming~\cite{Fang2018maxinfov2}.
Here, we propose a new approach --- we obtain the one-shot non-signalling assisted simulation cost by casting it as a resource dilution problem in our framework, which is a subject well studied in the context of resource theories~\cite{Brandao2011oneshot,Zhao2018oneshot,Liu2019oneshot,Yuan2019memory}.

Our results provide new perspective to the fundamental problems in quantum Shannon theory while lifting the resource theory of channels to effective tools to address concrete problems, and the arguments employed here are expected to be extendable to more generic settings thanks to the systematic nature of the resource-theoretic frameworks.

\textit{\textbf{Resource theory of communication.}}
---
Let $\mT(A,B)$ be the set of quantum channels with input system $A$ and output system $B$, while we omit the specification of input/output systems when it is clear from the context.  
We would like to quantitatively understand how useful the given channel is for communication tasks, and a reasonable as well as operationally motivated approach for this purpose is to take ``useless'' channels for communication as the set of free channels. 
For classical communication over quantum channels, a natural choice for useless channels are constant channels~\cite{Cooney2015strong,Berta2018amortized}, which map any state to some fixed state. 
Namely, we choose the set of free channels as  
\ba
 \mfF:=\lset \Xi\in\mT \sbar \exists \sigma\mbox{ s.t. }\Xi(\rho)=\sigma,\ \forall \rho \rset.
\ea
It is clear that constant channels are not able to transmit any information and thus suitable for our choice of free resources --- indeed a channel has zero classical capacity if and only if it is a constant channel. In other words, choosing a different set of free channels will allow for free communication, which is not appropriate to our setting.

Our goal is to gain ideas of usefulness of a given channel in this framework, which motivates us to quantify the resourcefulness of channels with respect to the set of constant channels.
To this end, we consider robustness of communication, the generalized robustness measure~\cite{Steiner2003generalized,Napoli2016robustness,regula_2018,Diaz2018usingreusing,Takagi2019operational,Takagi2019general,Liu2019channel} with respect to the set of constant channels defined for any channel $\mN$ as \footnote{We remark that \eqref{eq:robustness def} is always well-defined because the set of Choi matrices corresponding to the constant channels include a full-rank operator.}
 \ba
  R(\mN):= &\min\lset r\geq 0 \sbar \frac{\mN+r\mL}{1+r}\in \mfF,\ \mL\in\mT\rset,
 \label{eq:robustness def}
 \ea
which was introduced in Ref.~\cite{fang2018thesis}. 
It is also convenient to consider the max-relative entropy of communication:
\ba
 \mfD_{\max}(\mN):= &\min\lset s\sbar \mN\leq 2^s\mL,\ \mL\in\mfF\rset,
 \label{eq:Dmax def}
\ea
where the inequality is in terms of complete positiveness. 
Then, it is straightforward to see the relation $\mfD_{\max}(\mN)=\log(1+R(\mN))$. 
We discuss their properties in detail in Appendix~\ref{sec:app property}.

Besides the quantification of resources, another central theme that resource theories deal with is manipulation of resources. 
Since our resource objects are quantum channels, it is natural to consider channel transformations under superchannels~\cite{Chiribella2008architecture,Chiribella_2008transforming}. 
Let $\mS(\{A,B\},\{A',B'\})$ be the set of superchannels that map channels in $\mT(A,B)$ to channels in $\mT(A',B')$.
Of particular interest are channel transformations under ``free'' superchannels. 
The requirement for free superchannels is that they do not create resourceful channels out of free channels. 
Within this constraint, there is still much freedom in what additional constraints one should impose~\cite{Liu2019channel,Gour2019bipartite}.  
Here, we will take a similar approach to Ref.~\cite{Gour2019bipartite}, considering the maximal set of free superchannels (often called ``resource non-generating'') defined as~\footnote{Ref.~\cite{Gour2019bipartite} chooses the \textit{completely} resource non-generating superchannels, which do not create any resource even by applications on partial systems that free channels act on, as free superchannels. In our case where $\mfF$ is the set of constant channels, the definition in \eqref{eq:free superchannel condition} coincides with the set of completely resource non-generating superchannels.}
\ba
\mO_\mfF:=\lset \Theta\in\mS \sbar \Theta[\Xi]\in \mfF,\ \forall \Xi\in\mfF\rset  
\label{eq:free superchannel condition}. 
\ea

Although this type of ``maximal'' choice of free operations is motivated by a mathematical convenience and usually does not have a good characterization (e.g. separability preserving operations for bipartite entanglement~\cite{Brandao2011oneshot}, maximally incoherent operations for coherence~\cite{Aberg2006quantifying}), it turns out that in our case this choice of free superchannels exactly corresponds to the communication setting with non-signalling assistance, which connects the mathematical formulation of resource theory and communication tasks. (See Refs. \cite{DiVincenzo1999entanglement,Chitambar2016assisted,Morris2018assited} for
other resource assisted tasks considered in different settings.)
More precisely, consider a channel transformation by a non-signalling bipartite channel $\Pi_{\rm NS}:A_i B_i\rightarrow A_oB_o$ converting channel $\mN\in\mT(A_o,B_i)$ to another channel $\mN'\in\mT(A_i,B_o)$ where $\Pi_{\rm NS}$ satisfies the non-signalling conditions~\cite{Duan2016nosignalling}
\ba
\Tr_{A_o}\Pi_{\rm NS}(\rho_{A_i}^{(0)}\otimes\rho_{B_i})=\Tr_{A_o}\Pi_{\rm NS}(\rho_{A_i}^{(1)}\otimes\rho_{B_i}) \label{eq:AtoB non-signalling}\\
\Tr_{B_o}\Pi_{\rm NS}(\rho_{A_i}\otimes\rho_{B_i}^{(0)})=\Tr_{B_o}\Pi_{\rm NS}(\rho_{A_i}\otimes\rho_{B_i}^{(1)}) \label{eq:BtoA non-signalling}
\ea
for any state $\rho_{A_i}$, $\rho_{B_i}$, and any pair of states $\{\rho_{A_i}^{(j)}\}_{j=0}^1$, $\{\rho_{B_i}^{(j)}\}_{j=0}^1$.
The subscript $i(o)$ in $A_{i(o)}$, $B_{i(o)}$ indicates that the systems are input (output) systems of the bipartite channel $\Pi_{\rm NS}$.
Eq.~\eqref{eq:BtoA non-signalling} ensures that the bipartite operation $\Pi_{\rm NS}$ is ``semicausal'' from $A$ to $B$~\cite{Beckman2001causal}, which is shown to be ``semilocalizable''~\cite{Eggeling2002semilocalizable} and
constructs a ``quantum comb'' with a causal order~\cite{Chiribella_2008transforming,Chiribella2008architecture} as shown in Fig.~\ref{fig:NS_assisted}.

\begin{figure}[htbp]
    \centering
    \includegraphics[scale=0.3]{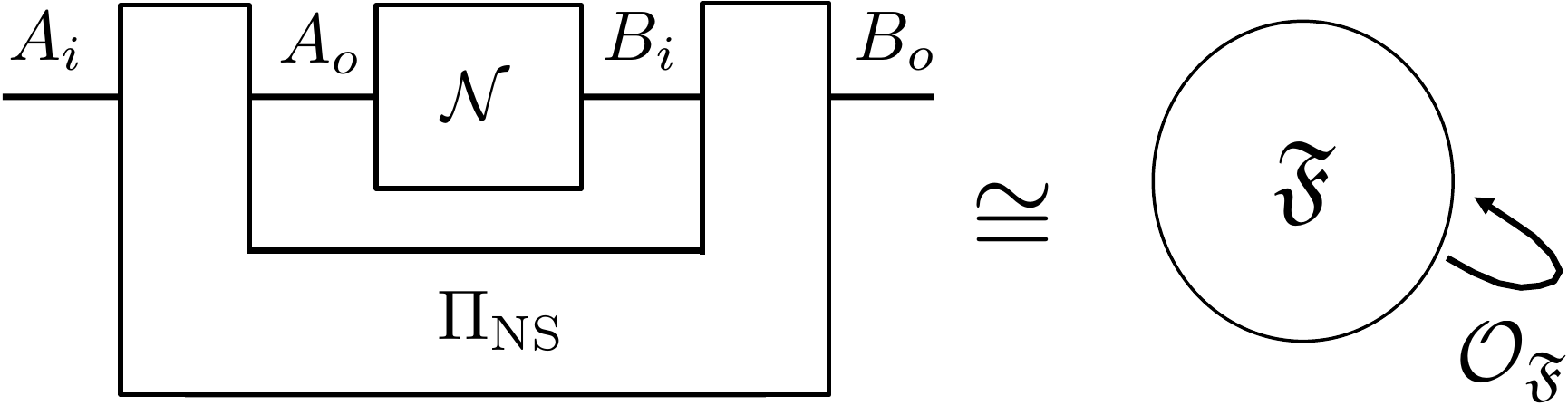}
    \caption{Non-signalling bipartite channel $\Pi_{\rm NS}$ constructs a quantum comb (left). Proposition \ref{pro:FS = NS} shows that the non-signalling assisted channel transformation is equivalent to the maximal set of superchannels that preserves the set of constant channels (right).}
    \label{fig:NS_assisted}
\end{figure}

Let $\mO_{\rm NS}$ be the set of superchannels realized by non-signalling channels; i.e. $\mO_{\rm NS}:=\{\Theta\in\mS\,|\,\Theta[\mN]=\Pi_{\rm NS}\circ \mN,\ \forall \mN\in\mT\}$ where $\Pi_{\rm NS}$ is a non-signalling channel satisfying \eqref{eq:AtoB non-signalling} and \eqref{eq:BtoA non-signalling} and the concatenation ($\circ$) refers to the comb structure in Fig.~\ref{fig:NS_assisted}.
Then, we have the following identification between two sets of channel transformations. (Proof in Appendix~\ref{sec:FS=NS}.) 
\begin{pro} \label{pro:FS = NS}
The set of resource non-generating free superchannels coincides with that of non-signalling assisted channel transformations, i.e. $\mO_\mfF=\mO_{\rm NS}$. 
\end{pro}
This result shows that the maximal set of free superchannels, which is motivated by a mathematical convenience, is characterized by the ultimate bipartite correlation respecting the causality in the theory of relativity~\cite{Beckman2001causal} and thus gains a physical characterization. Furthermore, Proposition~\ref{pro:FS = NS} allows us to consider the non-signalling assisted channel coding as a channel transformation under the maximal set of free superchannels, which will be useful for later discussions.

Finally, we show the monotonicity property of a ``smoothed'' version of the max-relative entropy measure under free superchannels, ensuring it to be a valid resource monotone. Moreover, the monotonicity holds for general resource theories, in which our theory is included as a special case. 
Define the smooth max-relative entropy measure with respect to an arbitrary set of free channels $\mfF'$ as 
$\mfD_{\max,\mfF'}^\epsilon(\mN):=\min_{\|\mN'-\mN\|_\diamond\leq \epsilon} \mfD_{\max,\mfF'}(\mN')$ where $\mfD_{\max,\mfF'}(\mN')$ is the max-relative entropy measure with respect to $\mfF'$ defined analogously to \eqref{eq:Dmax def}, and $\|\cdot\|_\diamond$ is the diamond norm.
Then, we have the following general monotonicity property of the smoothed measure. (Proof in Appendix~\ref{sec:app converse alternative}.)
\begin{lem}\label{lem:monotonicity smoothed Dmax}
 Let $\mfF'$ be an arbitrary set of free channels, $\mO_{\mfF'}$ be the corresponding set of resource non-generating free superchannels. Then, for any channel $\mN$ and free superchannel $\Theta\in\mO_{\mfF'}$, it holds that $\mfD_{\max,\mfF'}^\epsilon(\mN)\geq\mfD_{\max,\mfF'}^\epsilon(\Theta[\mN])$ for $\epsilon\geq 0$.  
\end{lem}

In the next section, we see the relation between the resource measure developed here and operational advantage in state discrimination tasks, which is naturally connected to classical communication tasks we will discuss later. 

\textit{\textbf{Operational characterization of channel resources.}}
---
Characterizing an operational advantage enabled by the given resource is a central question in quantum information theory, and recent works have clarified that discrimination tasks are effective platforms to investigate for this purpose. 
More specifically, let $p_{\rm succ}(\mA,\Lambda,\{M_i\}):=\sum_i p_i\Tr[\Lambda(\sigma_i)M_i]$ be the average success probability for discriminating the given state ensemble $\mA=\{p_i,\sigma_i\}$ with the action of channel $\Lambda$. Then, the following relation has been shown to hold for general resource theories.

\begin{lem}[\cite{Takagi2019general}] \label{lem:operational characterization}
 For any convex and closed set of free channels $\mfF'$, it holds that, for any channel $\mN\in\mT(A,B)$,
 \bann
  \max_{\mA,\{M_i\}}\frac{p_{\rm succ}(\mA, \id_E\otimes\mN,\{M_i\})}{\max_{\Xi\in\mfF'}p_{\rm succ}(\mA,\id_E\otimes\Xi,\{M_i\})} = 1+R_{\mfF'}(\mN)
 \eann
 where $\mA$ is the state ensemble defined on the system $EA$ with $E$ being some quantum system, and $R_{\mfF'}(\mN)$ is the generalized robustness defined with respect to the set of free channels $\mfF'$. 
\end{lem}

Lemma \ref{lem:operational characterization} was shown aiming to characterize the operational advantage enabled by the given resourceful channel with respect to free resources in terms of the performance of state discrimination tasks at a high level of generality. However, to obtain such a general result, the same measurement strategy $\{M_i\}$ needs to be employed on comparing the resourceful case and resourceless case, which makes the task a little artificial.
Interestingly, the simple structure of our theory allows for a more natural and convenient form in the left-hand side of the above result --- in particular, the optimization over measurements can be taken separately in the numerator and denominator, and the denominator can be reduced to the form relevant to our communication setting. 
Let $p_{\rm succ}(\mA,\mN):=\max_{\{M_i\}}p_{\rm succ}(\mA,\mN,\{M_i\})$ be the optimal success probability, and $p_{\rm guess}(\mA):=\max_i p_i$ be the success probability for the best random guess.
Then, we obtain the following. (Proof in Appendix~\ref{sec:app discrimination proof}.)
\begin{thm} \label{thm:discrimination robustness constant}
For any channel $\mN\in\mT(A,B)$, we get
\ba
 \max_{\mA\in \mathfrak{E}}\frac{p_{\rm succ}(\mA,\id_E\otimes\mN)}{p_{\rm guess}(\mA)} = 1+R(\mN)
 \label{eq:discrimination}
\ea
 where
$
\mathfrak{E}:=\lset\{p_i,\sigma_i^{EA}\}\sbar\Tr_A[\sigma_i^{EA}]=\Tr_A[\sigma_j^{EA}],\ \forall i,j\rset.
$ 
In particular, considering the case when the external system $E$ is one-dimensional, we have that for any ensemble $\mA$ supported on system $A$,
\ba
 \frac{p_{\rm succ}(\mA,\mN)}{p_{\rm guess}(\mA)} \leq 1+R(\mN).
 \label{eq:discrimination2}
\ea
\end{thm}

Now, we are in a position to consider communication tasks, for which we employ the relation between resource measure and state discrimination. 

\textit{\textbf{Converse bound for the assisted capacity.}}
---
Let $M$ be the number of classical messages Alice tries to send to Bob. 
Also, let $A_i$ and $B_o$ be the $m$-dimensional message space and $\Theta\in\mS(\{A_o,B_i\},\{A_i,B_o\})$ be a superchannel that is constructed by a non-signalling bipartite channel as in Fig.~\ref{fig:NS_assisted}, i.e. $\Theta\in\mO_{\rm NS}$. (We also call this \textit{non-signalling assisted code} in the context of message transmission.)
For a given superchannel, we write the average error probability of decoding as $\varepsilon[\Theta,\mN]:=1-\frac{1}{M}\sum_{m=0}^{M-1}\bra{m}\Theta[\mN](\dm{m})\ket{m}$.
Then, the non-signalling assisted one-shot classical capacity with error $\epsilon$ is defined as~\cite{Wang2018converse} 
\ba
 C_{{\rm NS},(1)}^{\epsilon}(\mN):=\sup_\Theta\lset \log M\sbar \varepsilon[\Theta,\mN]\leq \epsilon\rset.
\ea

Using our framework and above result, we can concisely show a converse bound for the one-shot non-signalling assisted classical capacity. 

\begin{thm} \label{thm:one-shot bound}
For $\delta\geq 0$ and $0\leq\epsilon< 1-\delta/2$, it holds that 
\ba
 C_{{\rm NS},(1)}^{\epsilon}(\mN)\leq \mfD_{\max}^\delta(\mN)+\log\left(\frac{1}{1-\epsilon-\delta/2}\right)
\ea
\end{thm}

\begin{proof}

Let $\mA=\{1/M,\dm{m}\}_{m=0}^{M-1}$ and $\Theta$ be a non-signalling assisted code achieving $\varepsilon[\Theta,\mN]= 1-p_{\rm succ}(\mA,\Theta[\mN])\leq \epsilon$. 
Using Proposition \ref{pro:FS = NS}, we have that $\Theta\in\mO_\mfF$. 
Now, take the channel $\mL$ with $\|\mL-\Theta[\mN]\|_\diamond\leq \delta$ that satisfies $\mfD_{\max}(\mL)=\mfD_{\max}^\delta(\Theta[\mN])$. Using \eqref{eq:discrimination2} in Theorem \ref{thm:discrimination robustness constant} and $p_{\rm guess}(\mA)=1/M$, we get 
\ba
p_{\rm succ}(\mA,\mL)\leq \frac{1+R(\mL)}{M}=2^{\mfD_{\max}^\delta(\Theta[\mN])}/M
\label{eq:error succ uniform classical alt}
\ea
where we used $\log(1+R(\mL))=\mfD_{\max}(\mL)=\mfD_{\max}^\delta(\Theta[\mN])$.
Then, we use the following simple lemma, which we show in Appendix~\ref{sec:guess diamond}.

\begin{lem} \label{lem:guess diamond}
For any state ensemble $\mA$ and channels $\mL,\mM\in \mT(A,B)$, we have
$\left|p_{\rm succ}(\mA,\mM)-p_{\rm succ}(\mA,\mL)\right| \leq \frac{1}{2}\|\mM-\mL\|_\diamond$.
\end{lem}

Applying  Lemma~\ref{lem:guess diamond} and the monotonicity of the smooth max-relative entropy measure (Lemma~\ref{lem:monotonicity smoothed Dmax}) to \eqref{eq:error succ uniform classical alt}, we reach $p_{\rm succ}(\mA,\Theta[\mN])-\delta/2\leq 2^{\mfD_{\max}^\delta(\mN)}/M.$
The proof is completed by taking the logarithm on both sides and using $p_{\rm succ}(\mA,\Theta[\mN])=1-\varepsilon[\Theta,\mN]\geq 1-\epsilon$.

\end{proof}

This result establishes a fundamental connection between the resourcefulness of the channel quantified in the resource-theoretic framework and its operational capability as a communication channel.

We remark that related bounds have been presented in Refs.~\cite{Datta2013oneshot,Matthews2014finite}.
An advantage of our result over their bounds is that, besides the simplicity of its proof, it naturally leads to the strong converse property as we shall see below.  
Note also that by combining the relations presented in Refs.~\cite{Matthews2014finite,Wang2019converse,Wang2019asymmetric} one can alternatively reach the same one-shot bound in Theorem~\ref{thm:one-shot bound}~\cite{Fang2019comment}.


\textit{\textbf{Strong converse property for the assisted capacity.}}
---
Next, we discuss the advantage of Theorem \ref{thm:one-shot bound} in the asymptotic setting. For this aim, we take into account the situation where multiple copies of the channel are in use. Consider a sequence of the message size $M^{(n)}$ and non-signalling assisted codes $\Theta^{(n)}$. 
Then, we can define the non-signalling assisted classical capacity as the maximum rate of the message transmission with vanishing error in the asymptotic limit:~\footnote{Note that when it comes to the asymptotic setting, non-signalling assisted quantum capacity equals to a half of the non-signalling assisted classical capacity due to the interconvertibility between the two communication settings via quantum teleportation and superdense coding protocols.}
\bal
 C_{\rm NS}(\mN):=\sup_{\{\Theta^{(n)}\}}\lset \underline{\lim} \frac{\log M^{(n)}}{n}\sbar \lim_{n\to\infty} \varepsilon[\Theta^{(n)},\mN^{\otimes n}]=0\rset.
 \label{eq:NS assisted capacity def}
\eal
We also introduce the non-signalling assisted strong converse capacity $C_{\rm NS}^\dagger(\mN)$ by replacing $\lim_{n\to\infty} \varepsilon[\Theta^{(n)},\mN^{\otimes n}]=0$ in \eqref{eq:NS assisted capacity def} with $\lim_{n\to\infty} \varepsilon[\Theta^{(n)},\mN^{\otimes n}]<1$.
By definition, it holds that $C_{\rm NS}(\mN)\leq C_{\rm NS}^\dagger(\mN)$ for any $\mN$. If $C_{\rm NS}(\mN)= C_{\rm NS}^\dagger(\mN)$ also holds, we say that the strong converse property holds. 

A direct proof of the strong converse property for the entanglement assisted capacity without employing the operational argument via quantum reverse Shannon theorem was reported in Ref.~\cite{Gupta2015strong}. 
There, they first put an upper bound for the decoding success probability in terms of variants of mutual information derived from the $\alpha$-sandwiched R\'enyi entropy~\cite{Muller2013renyi,Wilde2014strong} using the meta-converse bound~\cite{Matthews2014finite}. 
Then, they showed the additivity of the $\alpha$-mutual information for $\alpha\in(1,\infty)$ using the multiplicativity of completely bounded $p$-norms~\cite{Devetak2006multiplicativity}, which allowed them to connect the $\alpha$-mutual information to the usual mutual information, eventually proving the strong converse.   

Here, we see that Theorem~\ref{thm:one-shot bound} naturally allows for this type of direct proof of an even stronger claim, the strong converse property for the non-signalling assisted capacity, without delving into involved steps such as the ones in Ref.~\cite{Gupta2015strong}. 
The main idea is to combine Theorem~\ref{thm:one-shot bound} with the following asymptotic equipartition property~\cite{Fang2018maxinfov2},
\ba
\lim_{\delta\to 0}\lim_{n\to\infty}\frac{1}{n}\mfD_{\max}^\delta(\mN^{\otimes n})=I(\mN)
\label{eq:max info AEP}
\ea
where $I(\mN):= \max_{\ket{\psi}}I(\rho_{AB})$ and $\rho_{AB}:=\id\otimes\mN(\dm{\psi})$ is the channel mutual information.

\begin{cor} \label{pro:strong converse}
For any channel $\mN$, the strong converse property holds for the non-signalling assisted capacity, i.e. $C_{\rm NS}(\mN)=C_{\rm NS}^\dagger(\mN)$.
\end{cor}

\begin{proof}
Theorem~\ref{thm:one-shot bound} implies that for any $n$ and $0\leq\epsilon<1-\delta/2$,
\bann
 \frac{1}{n}C_{{\rm NS},(1)}^{\epsilon}(\mN^{\otimes n})\leq \frac{1}{n}\mfD_{\max}^\delta(\mN^{\otimes n})+\frac{1}{n}\log\left(\frac{1}{1-\epsilon-\delta/2}\right). 
\eann
Taking $\lim_{\delta\to 0}\lim_{n\to\infty}$ in both sides and using \eqref{eq:max info AEP}, we obtain $\lim_{n\to \infty}\frac{1}{n}C_{{\rm NS},(1)}^{\epsilon}(\mN^{\otimes n})\leq I(\mN)=C_{\rm EA}(\mN)$ for any $0\leq \epsilon <1$ where $C_{\rm EA}$ is the entanglement assisted classical capacity~\cite{Bennett2002capacity}.
This proves $C_{\rm NS}(\mN)\geq C_{\rm EA}(\mN) \geq C_{\rm NS}^\dagger(\mN)$, showing the strong converse property.
\end{proof}

This result provides a new perspective to the ultimate communication capability with many channel uses in terms of the operational characterization of channel resources.

\textit{\textbf{Channel simulation.}}
---
Proposition~\ref{pro:FS = NS} also allows us to identify the non-signalling assisted channel simulation with the resource dilution problem in our resource theory.  
Specifically, let $\id_k$ be the identity channel acting on $k$-dimensional Hilbert space. 
We ask the minimum size of the identity channel needed to realize the desired channel by free superchannels. 
To this end, we define the one-shot dilution cost for given channel $\mN$ and error $\epsilon$ as 
\bann
 C_{c,(1)}^{\epsilon}(\mN) := \min\lset k \sbar \exists \Theta\in\mO_\mfF\ {\rm s.t.}\ \|\Theta[\id_k]-\mN\|_\diamond\leq \epsilon\rset.
\eann

Then, we obtain the following. (Proof in Appendix~\ref{sec:app simulation cost}.)

\begin{thm}\label{thm:dilution}
\ba
 C_{c,(1)}^{\epsilon}(\mN) = \lceil 2^{\frac{1}{2}\mfD_{\max}^\epsilon(\mN)}\rceil.
\ea
\end{thm}

This result provides the generalized robustness/max-relative entropy measure with another operational meaning. 
As expected, our result coincides with the non-signalling assisted one-shot channel simulation cost obtained by a different approach~\cite{Fang2018maxinfov2}.
Since our method is based on a systematic resource theoretic treatment, it will provide a useful tool with wide applicability. 
We also remark that because of \eqref{eq:max info AEP}, the asymptotic cost is characterized by the mutual information of the channel, which makes the channel transformation reversible at the asymptotic limit. 
Our resource-theoretic treatment makes the comparison to other reversible theories clearer --- in our case the mutual information serves as the ``potential'' function that fully characterizes the resource transformability.

\textit{\textbf{Conclusions.}}
---
We introduced a resource theory of channels relevant to communication scenarios where the set of constant channels serves as the free channels. 
We considered channel transformation under the maximal set of free superchannels and found that such channel transformation coincides with that under non-signalling assistance.
Employing this identification, we applied our formalism to provide a converse bound for the one-shot non-signalling assisted classical capacity, which leads to the strong converse property for the non-signalling assisted capacity, as well as to obtain the one-shot channel simulation cost with non-signalling assistance by considering the resource dilution cost under free superchannels. 
Both of the quantities are characterized by the max-relative entropy measure with respect to our choice of free channels, endowing this measure with clear operational meanings. 

Our results indicate the further potential of resource theoretic framework as effective tools to solve concrete problems. In this respect, an interesting future direction will be to adopt our method to encompass other communication settings such as non-assisted classical/quantum communication and communication with restricted quantum measurements. 

\textit{Note added.}
---
Recently, we became aware that the latest update of Ref.~\cite{Fang2018maxinfov2} has obtained a similar relation to the one in Lemma~\ref{lem:monotonicity smoothed Dmax} for the case of non-signalling superchannels.

\textit{Acknowledgements}
---
We thank Kun Fang and Xin Wang for useful comments on the manuscript. 
R.T. acknowledges the support of NSF, ARO, IARPA, and Takenaka Scholarship Foundation.
M.H. was supported in part by JSPS Grant-in-Aid for Scientific Research (A) No.17H01280 and for Scientific Research (B) No.16KT0017, and Kayamori Foundation of Informational Science Advancement.

\bibliographystyle{apsrmp4-2}
\bibliography{myref}

\appendix
\widetext
\section{Properties of $\mfD_{\max}(\mN)$} \label{sec:app property}
First, it is straightforward to see that the max-relative entropy measure, as well as the robustness measure, inherits generic properties such as faithfulness and monotonicity under free operations, which can be proved in the same way as the case for states (see e.g.,~\cite{regula_2018}). 
Here, we elaborate on the other properties of this measure: additivity under tensor product, its tight upper bound, and the relation with max-information~\cite{Fang2018maxinfov2}.  

Recall the definition of the max-relative entropy between two states: 
$D_{\max}(\rho||\sigma):= \min\lset s \sbar \rho\leq 2^s\sigma\rset$. 
We first show the following additivity property.
\begin{pro}\label{pro:additivity}
$\mfD_{\max}(\mN_1\otimes\mN_2)=\mfD_{\max}(\mN_1)+\mfD_{\max}(\mN_2)$ for any channels $\mN_1$ and $\mN_2$.
\end{pro} 

\begin{proof}
Notice that 
\ba
 \mfD_{\max}(\mN)=\min_{\mM\in\mfF}D_{\max}(\mN||\mM)
 \label{eq:dmax def min}
\ea
where $D_{\max}(\mN||\mM)=D_{\max}(J_\mN||J_\mM)=\min\{s|J_\mN\leq 2^s J_\mM\}$ with $J_\mN$, $J_\mM$ being the Choi matrices for channels $\mM$, $\mN$. 
Let $\mM_1,\mM_2\in\mfF$ be the constant channels satisfying $\mfD_{\max}(\mN_1)=D_{\max}(\mN_1||\mM_1)$, $\mfD_{\max}(\mN_2)=D_{\max}(\mN_2||\mM_2)$.
Noting that $\mM_1\otimes\mM_2$ is also a constant channel, we get
\ba
 \mfD_{\max}(\mN_1\otimes\mN_2)\leq D_{\max}(\mN_1\otimes\mN_2||\mM_1\otimes\mM_2)=D_{\max}(\mN_1||\mM_1)+D_{\max}(\mN_2||\mM_2)=\mfD_{\max}(\mN_1)+\mfD_{\max}(\mN_2).
\ea
where in the first equality we used the additivity of $D_{\max}$ for product (Choi) states. 
To show the superadditivity, consider the dual form of the max-relative entropy measure that can be obtained by a standard technique of convex optimization~\cite{boyd_2004} (see also~\cite{Takagi2019general}). 
For any channel $\mN\in \mT(A,B)$, the max-relative entropy of communication is evaluated by
\ba
 \mfD_{\max}(\mN)&=&\max\lset \log\Tr[Y J_\mN] \sbar Y\geq 0,\ \Tr[Y(\mbI\otimes \sigma)]\leq 1,\ \forall \sigma\rset\\
 &=&\max\lset \log\Tr[Y J_\mN] \sbar Y\geq 0,\ \|Y_B\|_\infty\leq 1\rset \label{eq:dmax dual2}
\ea
where $J_\mN$ is the Choi matrix for $\mN$, and $Y_B:=\Tr_A[Y]$. 
Let $Y_1$ and $Y_2$ be optimal solutions for $\mN_1$ and $\mN_2$. 
Due to the multiplicativity of the operator norm, $\|Y_{1,B}\|_\infty\leq 1$ and $\|Y_{2,B}\|_\infty\leq 1$ imply $\|\Tr_A[Y_1\otimes Y_2]\|_\infty=\|Y_{1,B}\otimes Y_{2,B}\|_\infty=\|Y_{1,B}\|_\infty\|Y_{2,B}\|_\infty\leq 1$. Thus, $Y_1\otimes Y_2$ is a valid solution for $\mN_1\otimes \mN_2$, which proves $\mfD_{\max}(\mN_1\otimes \mN_2)\geq \mfD_{\max}(\mN_1)+\mfD_{\max}(\mN_2)$. 
\end{proof}

From \eqref{eq:dmax def min}, it can be seen that the max-relative entropy of communication coincides with the conditional min-entropy of the Choi matrix of the channel~\cite{Duan2016nosignalling}.
Hence, Proposition~\ref{pro:additivity} can also be understood as the additivity property of the conditional min-entropy shown in Ref.~\cite{Konig2009operational}.

It is also worth noting that the max-relative entropy of communication coincides with the quantity known as max-information for channels introduced in Ref.~\cite{Fang2018maxinfov2}. 
This has been shown in Ref.~\cite{fang2018thesis} --- we attach the proof here for completeness.

\begin{lem} \label{lem:dmax Imax}
 Consider the max-information defined for any channel $\mN\in\mT(A,B)$ by $I_{\max}(\mN):=I_{\max}(\id_A\otimes\mN(\Phi_{AA}))$ where $\Phi_{AA}$ is the maximally entangled state on $A$ and $A'\cong A$ (which we simply write $AA$), and  $I_{\max}(\rho_{AB}):=\min_{\sigma}D_{\max}(\rho_{AB}||\rho_A\otimes \sigma)$ is the max-information defined for states~\cite{Berta2011reverse}. Then, it holds that $\mfD_{\max}(\mN) = I_{\max}(\mN)$.
\end{lem}

\begin{proof}
Rewriting the definition of max-relative entropy, we get 
\ba
 \mfD_{\max}(\mN)&=&\min_{\mM\in\mfF} D_{\max}(\mN||\mM)\\
 &=& \min_{J_\mM\in\mfF} D_{\max}(J_\mN||J_\mM)\\
 &=& \min_{\sigma} D_{\max}(J_\mN||\mbI_A\otimes \sigma)\\
 &=& \min_{\sigma} D_{\max}(\id\otimes \mN (\Phi)||\mbI_A/d_A\otimes \sigma)\\
 &=& I_{\max}(\mN).
\ea
\end{proof}

Next, we present the tight upper bound for this measure and show that it is achieved by the reversible channels. 
\begin{pro} \label{pro:dmax upper bound}
Let $\mN\in\mT(A,B)$ and $d_A$ be the dimension of the underlying Hilbert space in system $A$. Then, it holds that $\mfD_{\max}(\mN)\leq 2\log d_A$, and the equality is achieved if and only if $\mN$ is reversible.  
\end{pro}

\begin{proof}
We first argue that \eqref{eq:dmax dual2} can be rewritten (in terms of the generalized robustness) as
\ba
 1+R(\mN)&=&\max \lset \Tr[J_\mN Y_{AB}] \sbar \Tr_{A}[Y_{AB}]\leq \mbI_B,\ Y_{AB}\geq 0\rset\\
 &=& \max \lset \Tr[J_\mN Y_{AB}] \sbar \Tr_{A}[Y_{AB}]= \mbI_B,\ Y_{AB}\geq 0\rset
 \label{eq:robustness dual equality}
\ea
where the second equality is shown as follows: if $\Tr_{A}[Y_{AB}]\neq \mbI_B$, there exists a positive semidefinite operator $Q_B\geq 0$ such that $\mbI_B-\Tr_{A}[Y_{AB}]=Q_B$. 
Then, it is straightforward to check that another operator $Y_{AB}':= Y_{AB}+\sigma\otimes Q_B$ where $\sigma\geq 0$ and $\Tr[\sigma]=1$ satisfies $\Tr_{A}[Y_{AB}']=\mbI_B$, $Y_{AB}'\geq 0$, and $\Tr[J_N Y_{AB}']\geq \Tr[J_N Y_{AB}]$ since $J_N\geq 0$ and $Y_{AB}'-Y_{AB}=\sigma\otimes Q_B \geq 0$. 

Let $Y_{AB}$ be a dual solution in \eqref{eq:robustness dual equality}. 
Noting $\Tr_A Y_{AB} = \mbI_B$, we take a unital map $\mE:=J^{-1}(Y_{AB}):A\to B$, which has $Y_{AB}$ as its Choi matrix. Furthermore, let $\mF = \mE^\dagger:B
\to A$ be the adjoint quantum channel of $\mE$. Then writing $\Gamma_{AA}:=d_A\dm{\Phi_{AA}}=\sum\ketbra{ii}{jj}$ be the unnormalized maximally entangled state acting on $A$ and $A'\cong A$ (which we simply write $A$), we get 
\begin{align}
1+R(\mN) &= \max_{Y_{AB} \geq 0,\Tr_AY_{AB} = \mbI_B} \Tr\left[J_\mN Y_{AB}\right] \\
&=  \max_{\mE_{A\to B}}
    \Tr\left[ J_\mN (\id_{A}\otimes\mE_{A\to B}) (\Gamma_{AA}) \right] \\
&=  \max_{\mE_{A\to B}}
    \Tr\left[ (\id_{A}\otimes\mE_{A\to B})^\dagger (J_{\mN})\, \Gamma_{AA} \right] \\
&=  \max_{\mF_{B\to A}}
    \Tr\left[ (\id_{A}\otimes\mF_{B \to A})(J_{\mN})\, \Gamma_{AA} \right] \\
&=  d_A\max_{\mF_{B\to A}}
    \bra{\Phi_{AA}}\left[(\id_{A}\otimes\mF_{B \to A})(J_{\mN})\right]\ket{\Phi_{AA}} \\
&=  d_A^2\max_{\mF_{B\to A}}
    \bra{\Phi_{AA}}\left[(\id_{A}\otimes\mF_{B \to A})(\id_{A}\otimes\mN_{A \to B})(\Phi_{AA})\right]\ket{\Phi_{AA}} \\
&=  d_A^2\max_{\mF_{B\to A}}\bra{\Phi_{AA}}\mF\circ\mN(\Phi_{AA})\ket{\Phi_{AA}}.
    \label{eq:optimal-overlap}
\end{align}

It is now clear from the expression that we have $1+R(\mN)\leq d_A^2$, and the maximum is achieved if and only if $\mN$ is reversible.

\end{proof}

Note that a compatible upper bound can be deduced by the upper bound on the max-information of the states, $I_{\max}(\rho_{AB})\leq 2\log\min\{|A|,|B|\}$~\cite{Berta2011reverse}, together with Lemma~\ref{lem:dmax Imax}.
The above argument further provides an operational meaning to $R(\mN)$ as the ability of $\mN$ to preserve the maximally entangled state in the sense that once $\Phi_{AA}$ is destroyed by $\mN$, how well we can recover $\Phi_{AA}$ from the destroyed state.

\section{Proof of Proposition \ref{pro:FS = NS}} \label{sec:FS=NS}

\begin{proof}
To see $\mO_{\rm NS}\subseteq \mO_{\mfF}$, suppose to the contrary that there exists $\Pi_{\rm NS}$ that maps constant channel $\mN_\sigma\in\mT(A_o,B_i)$, which always outputs a fixed state $\sigma$, to some non-constant channel $\mM\in\mT(A_i,B_o)$. 
Since $\mM$ is non-constant, there exists a pair of states $\rho_{A_i}^{(0)}$ and $\rho_{A_i}^{(1)}$ such that $\mM(\rho_{A_i}^{(0)})\neq\mM(\rho_{A_i}^{(1)})$. 
However, this implies that $\Tr_{A_o}\Pi_{\rm NS}(\rho_{A_i}^{(0)}\otimes \sigma)\neq \Tr_{A_o}\Pi_{\rm NS}(\rho_{A_i}^{(1)}\otimes \sigma)$, which violates \eqref{eq:AtoB non-signalling}. Thus, it must be the case that $\mO_{\rm NS}\subseteq \mO_{\mfF}$.

To see the other inclusion $\mO_{\mfF}\subseteq \mO_{\rm NS}$, note first that because of the causal structure equipped with $\mO_\mfF$ as a set of superchannels, $B\rightarrow A$ non-signalling condition \eqref{eq:BtoA non-signalling} is always ensured.
Thus, it suffices to show that any superchannel realized by an $A\rightarrow B$ signalling operation is also outside of $\mO_\mfF$.
In particular, we show that whenever the bipartite operation is $A\rightarrow B$ signalling, there always exists a constant channel that is transformed to non-constant channel by this transformation. 
Let $\Pi$ be a $A\rightarrow B$ signalling operation that violates \eqref{eq:AtoB non-signalling}.  
Let $\rho_{A_i}^{(0)}$, $\rho_{A_i}^{(1)}$, and $\sigma_{B_i}$ be the states such that $\Tr_{A_o}\Pi(\rho_{A_i}^{(0)},\sigma_{B_i})\neq \Tr_{A_o}\Pi(\rho_{A_i}^{(1)},\sigma_{B_i})$. 
Consider the constant channel $\mN_{\sigma_{B_i}}\in\mT(A_o,B_i)$, which always outputs $\sigma_{B_i}$, and let $\mN'$ be the channel transformed from $\mN_{\sigma_{B_i}}$ by this operation. 
Then, we have $\mN'(\rho_{A_i}^{(0)})\neq\mN'(\rho_{A_i}^{(1)})$, which means that $\mN'$ is not a constant channel. 
Hence, any superchannel outside $\mO_{\rm NS}$ cannot be a member of $\mO_\mfF$, concluding the proof. 

\end{proof}

\section{Proof of Lemma \ref{lem:monotonicity smoothed Dmax}} \label{sec:app converse alternative}

\begin{proof}

We first show that the diamond norm is contractive under any action of superchannel. Let $\mN,\mM\in\mT(A,B)$ be two channels and $\Xi\in\mS(\{A,B\},\{A',B'\})$ be a superchannel defined by $\Xi[\mN]=\mE_{\rm post}\circ(\id_E\otimes\mN)\circ\mE_{\rm pre}$. Also, let $\tilde{\rho}$ be a state on system $RA'$ that achieves the diamond norm of $\Xi[\mN]-\Xi[\mM]$. Then,
\ba
 \|\Xi[\mN]-\Xi[\mM]\|_\diamond &=& \|\id_R\otimes\Xi[\mN](\tilde{\rho})-\id_R\otimes\Xi[\mM](\tilde{\rho})\|_1\\
 &=& \|\id_R\otimes\left[\mE_{\rm post}\circ(\id_E\otimes\mN)\circ\mE_{\rm pre}\right](\tilde{\rho})-\id_R\otimes\left[\mE_{\rm post}\circ(\id_E\otimes\mM)\circ\mE_{\rm pre}\right](\tilde{\rho})\|_1\\
 &\leq& \|\id_{RE}\otimes\mN(\tilde{\rho}')-\id_{RE}\otimes\mM(\tilde{\rho}')\|_1\\
 &\leq& \|\mN-\mM\|_\diamond
\ea
where on the first inequality we set $\tilde{\rho}':=\id_R\otimes\mE_{\rm pre}(\tilde{\rho})$ and also used the contractivity of the trace norm under CPTP maps. 

Let us now take a channel $\tilde{\mN}$ such that $\|\tilde{\mN}-\mN\|_\diamond\leq\epsilon$ and $\mfD_{\max,\mfF'}(\tilde{\mN})=\mfD_{\max,\mfF'}^\epsilon(\mN)$. We then get
\ba
 \mfD_{\max,\mfF'}^\epsilon(\mN) = \mfD_{\max,\mfF'}(\tilde{\mN}) \geq 
 \mfD_{\max,\mfF'}(\Theta[\tilde{\mN}]) \geq 
  \mfD_{\max,\mfF'}^\epsilon(\Theta[\mN])
\ea
where the first inequality is due to the monotonicity of $\mfD_{\max,\mfF'}$ under free superchannels, and the second inequality is because $\|\Theta[\tilde{\mN}]-\Theta[\mN]\|_\diamond\leq \|\tilde{\mN}-\mN\|_\diamond\leq \epsilon$ due to the contractivity of the diamond norm under superchannels shown above.  
\end{proof}

As can be seen from the proof, the monotonicity can be generalized to any other smoothing with respect to a contractive distance measure.


\section{Proof of Theorem \ref{thm:discrimination robustness constant}} \label{sec:app discrimination proof}

\begin{proof}

We first show that ${\rm l.h.s.}\leq {\rm r.h.s.}$ in \eqref{eq:discrimination}. To this end, we show a more general statement which holds for general choice of free channels $\mfF'$: for any state ensemble $\mA$,
\bal
  \frac{p_{\rm succ}(\mA,\id_E\otimes\mN)}{\max_{\Xi\in\mfF'}p_{\rm succ}(\mA,\id_E\otimes\Xi)} \leq 1+R_{\mfF'}(\mN)
  \label{eq:discrimination general}.
\eal
By definition of $R_{\mfF'}(\mN)$, there exists a free channel $\Xi\in\mfF'$ and some channel $\mL$ such that $\mN = (1+R_{\mfF'}(\mN))\Xi - R_{\mfF'}(\mN)\mL$. 
Let $\{M_i'\}$ be the optimal measurement for the numerator of the l.h.s in \eqref{eq:discrimination general}. Then, for any ensemble $\mA=\{p_i,\sigma_i^{EA}\}$,
\bal
 p_{\rm succ}(\mA,\id_E\otimes\mN) &= \sum_j p_j \Tr[\id_E\otimes\mN (\sigma_j^{EA}) M_j'] \\
 &= (1+R_{\mfF'}(\mN))\sum_j p_j \Tr[\id_E\otimes\Xi(\sigma_j^{EA}) M_j'] - R_{\mfF'}(\mN)\sum_j p_j \Tr[\id_E\otimes\mL(\sigma_j^{EA}) M_j']\\
 &\leq (1+R_{\mfF'}(\mN))\sum_j p_j \Tr[\id_E\otimes\Xi(\sigma_j^{EA}) M_j'] \\
 &\leq (1+R_{\mfF'}(\mN))\max_{\{M_i\}} \max_{\Xi\in\mfF'}\sum_j p_j \Tr[\id_E\otimes\Xi(\sigma_j^{EA}) M_j]\\
 &= (1+R_{\mfF'}(\mN))\max_{\Xi\in\mfF'} p_{\rm succ}(\mA,\id_E\otimes\Xi),
 \label{eq:general upper bound proof}
\eal
which proves \eqref{eq:discrimination general}. 
In addition, when $\mA\in \mathfrak{E}$ and $\mfF'$ is the set of constant channels, the denominator of the l.h.s. of \eqref{eq:discrimination general} becomes 
\bal
 \max_{\Xi\in\mfF'}p_{\rm succ}(\mA,\id_E\otimes\Xi) &= \max_{\{M_i\}}\max_{\tau}\sum_j p_j\Tr[\sigma_j^E\otimes\tau M_j]\\
 &= p_{\rm guess}(\mA)
 \label{eq:p guess}
\eal 
where $\sigma_j^E=\Tr_A[\sigma_j^{EA}]$, and the second equality is due to the condition on $\mathfrak{E}$. 
This concludes the first part of the proof of \eqref{eq:discrimination}. 

To show the other inequality, take $E\cong A$ and consider the ensemble $\mA'$ defined on the system $EA$.  
Take $p_j=1/{d_A^2}$, $\sigma_j^{EA}=(P_j\otimes\id_A )\Phi^{EA}(P_j\otimes\id_A )$ for $j=0,\dots,d_A^2-1$ where $\Phi^{EA} = \frac{1}{d_A}\sum_{ik}\ketbra{ii}{kk}$ and $P_j$ is the $j$\,th Pauli operator. 
Let us also consider the operator $M_j'=\frac{1}{d_A}(P_j\otimes\id_B)Y_{EB}(P_j\otimes\id_B)$ where $Y_{EB}$ is an optimal solution in \eqref{eq:robustness dual equality}. 
One can check that $\{M_j'\}$ satisfies the condition for a POVM measurement since clearly $M_j'\geq 0$ because $Y_{EB}\geq 0$ and $P_j$ are unitary, and $\sum_j M_j' = \mbI_{E}\otimes \mbI_{B}$ since the random application of the Pauli operators serves as the completely depolarizing channel, which gives $\sum_j (P_j\otimes \id_B)Y_{EB}(P_j\otimes \id_B)=d_A\mbI_E\otimes \Tr_E[Y_{EB}]$ and $\Tr_E[Y_{EB}]=\mbI_B$ due to \eqref{eq:robustness dual equality}.
Defining $\mP_j(\cdot):=P_j\cdot P_j$, we obtain
\ba
 p_{\rm succ}(\mA',\id_{E}\otimes\mN ,\{M_j'\})&=& \sum_{j=0}^{d_A^2-1}\frac{1}{d_A^2}\Tr\left[\mP_j\otimes\mN \left(\Phi^{EA}\right) \frac{1}{d_A}\mP_j\otimes\id_B(Y_{EB})\right]\\
 &=& \sum_{j=0}^{d_A^2-1}\frac{1}{d_A^3}\Tr\left[\id_{E}\otimes\mN \left(\Phi^{EA}\right) Y_{EB}\right] \\
 &=& \frac{1}{d_A^2}\Tr\left[J_\mN Y_{EB}\right]\\
 &=& \frac{1+R(\mN)}{d_A^2}
\ea
where in the second equality Pauli operators are canceled out in the middle and we also used the cyclic property of the trace, and in the third equality we used $\frac{1}{d_A}J_\mN=\id_{E}\otimes\mN(\Phi^{EA})$. 
Since $\mA'\in \mathfrak{E}$, together with \eqref{eq:p guess} and $p_{\rm guess}(\mA')=\frac{1}{d_A^2}$, we get 
\ba
 \max_{\mA\in\mathfrak{E}}\frac{p_{\rm succ}(\mA,\id_{E}\otimes\mN)}{p_{\rm guess}(\mA)}\geq \frac{p_{\rm succ}(\mA',\id_{E}\otimes\mN)}{p_{\rm guess}(\mA')} = 1+R(\mN),
\ea
which completes the proof of \eqref{eq:discrimination}.

Finally, \eqref{eq:discrimination2} can be shown by following the same steps as \eqref{eq:general upper bound proof} while replacing $\id_E\otimes \mN$ ($\id_E\otimes \Xi$) with $\mN$ ($\Xi$) and noting $\max_{\Xi\in\mfF}p_{\rm succ}(\mA,\Xi)=p_{\rm guess}(\mA)$.

\end{proof}

\section{Proof of Lemma \ref{lem:guess diamond}} \label{sec:guess diamond}
\begin{proof}
We assume $p_{\rm succ}(\mA,\mM)\geq p_{\rm succ}(\mA,\mL)$ without loss of generality. Then,   
 \ba
  \left|p_{\rm succ}(\mA,\mM)-p_{\rm succ}(\mA,\mL)\right|&=& \max_{\{M_i'\}}\sum_j p_j \Tr[M_j' \mM(\sigma_j)]- \max_{\{M_i\}}\sum_j p_j \Tr[M_j \mL(\sigma_j)]\\
  &\leq& \max_{\{M_i'\}}\sum_j p_j \Tr[M_j' (\mM-\mL)(\sigma_j)]\\
  &\leq& \sum_j p_j \frac{1}{2}\|\mM-\mL\|_\diamond = \frac{1}{2}\|\mM-\mL\|_\diamond
 \ea
where on the first inequality we set $\{M_i\}$ to be the optimal measurement for the first term, and on the second inequality we used that for any state $\rho$ and POVM element $M_j$, it holds that
\ba
\|\mM-\mL\|_\diamond\geq \|(\mM-\mL)(\rho)\|_1 = \max_{0\leq P \leq \mbI}2\Tr[P(\mM-\mL)(\rho)]\geq 2\Tr[M_j(\mM-\mL)(\rho)].
\ea
\end{proof}

\section{Proof of Theorem~\ref{thm:dilution}} \label{sec:app simulation cost}

\begin{proof}
Let $\mN'$ be a channel that satisfies $\mfD_{\max}(\mN')=\mfD_{\max}^\epsilon(\mN)$ and $\|\mN'-\mN\|_\diamond\leq\epsilon$. If $k^2\geq 2^{\mfD_{\max}^\epsilon(\mN)}$, there exists a channel $\mL$ such that $\frac{\mN' + (k^2-1)\mL}{k^2}\in \mfF$. Consider the following superchannel $\Theta$ defined as
\ba
 \Theta[\Lambda]&:=&\frac{\Tr[J_{\id_k}J_\Lambda]}{k^2}\mN' + \frac{\Tr[(k\mbI-J_{\id_k})J_\Lambda]}{k^2}\mL 
\ea
where $J_\Lambda$ denotes the Choi matrix for channel $\Lambda$. 
It can be seen that $\Theta\in\mO_\mfF$ by considering a free channel $\Xi\in\mfF$ such that it outputs a fixed state $\Xi(\cdot)=\tau$. Then, we have $\frac{\Tr[J_{\id_k}J_\Xi]}{k^2}=\Tr[\Phi_k (\frac{\mbI}{k}\otimes\tau)]=\frac{1}{k^2}$, which ensures $\Theta[\Xi]\in\mfF$ due to the definition of $\mL$. 
Since $\Theta[\id_k]=\mN'$, this specific construction achieves the desired transformation, and thus $C_{c,(1)}^{\epsilon}(\mN)\leq \lceil2^{\frac{1}{2}\mfD_{\max}^\epsilon(\mN)}\rceil$. 

Suppose, on the other hand, there exists $\Theta\in\mO_\mfF$ such that $\Theta[\id_k]=\mN'$ with $\|\mN'-\mN\|_\diamond\leq\epsilon$. Then, we have
$\mfD_{\max}^\epsilon(\mN)\leq \mfD_{\max}(\mN')=\mfD_{\max}(\Theta[\id_k])\leq\mfD_{\max}(\id_k) = 2 \log k$
where we used the monotonicity of $\mfD_{\max}$ in the second inequality and Proposition~\ref{pro:dmax upper bound} in the last equality.
This proves $2^{\frac{1}{2}\mfD_{\max}^\epsilon(\mN)} \leq C_{c,(1)}^{\epsilon}(\mN)$, which concludes the proof.
\end{proof}

Note that the same result holds with respect to any distance measure as long as the same distance measure is used for the smoothing of the max-relative entropy measure.

We also remark that a similar argument has been employed in~\cite{Yuan2019memory} where they considered a resource theory of quantum memory and max-relative entropy measure associated with it.
A couple of important differences are: 1) Our measure is efficiently computable by SDP unlike the one in~\cite{Yuan2019memory}, for which only computable bounds can be obtained. 2) The robustness measure considered in~\cite{Yuan2019memory} is the \textit{standard} robustness~\cite{Vidal1999robustness,regula_2018} where one is restricted to mix free channels to make the mixture with the given channel free.
It can be seen that in our case the standard robustness diverges for any resourceful channel, which is a generic feature observed for affine resource theories~\cite{gour_2017}.

\end{document}